\documentclass[unsortedaddress,nofootinbib,reprint]{revtex4-1}  

\usepackage[utf8]{inputenc}
\usepackage{amssymb,latexsym}
\usepackage{amsbsy} 
\usepackage{amsmath,amsthm}

\usepackage[shortlabels]{enumitem}
\usepackage{hyperref}

\usepackage{xcolor}

\theoremstyle{plain}
\newtheorem{theorem}{Theorem}

\newtheorem{lemma}[theorem]{Lemma}
\newtheorem{corollary}[theorem]{Corollary}

\theoremstyle{remark}
\newtheorem{remark}{Remark}
\newtheorem{assumption}{Assumption}

\theoremstyle{definition}
\newtheorem{definition}[theorem]{Definition}

\newcommand{\R}{{\mathbb R}}

\renewcommand{\i}{i}
\renewcommand{\d}{\mathrm{d}}

\begin{document}

\title{Unique Continuation for the Magnetic Schrödinger Equation}

\author{Andre Laestadius}

\email{andre.laestadius@kjemi.uio.no}
\affiliation{
	Hylleraas Centre for Quantum Molecular Sciences, Department of Chemistry, University of Oslo, P.O.~Box 1033 Blindern, N-0315 Oslo, Norway}

\author{Michael Benedicks}

\affiliation{Department of Mathematics, Uppsala University, Box 480, 751 06 Uppsala, Sweden}

\author{Markus Penz}

\affiliation{
Max Planck Institute for the Structure and Dynamics of Matter, Luruper Chaussee 149, 22761 Hamburg, Germany}

\date{\today}

\begin{abstract}
The unique-continuation property from sets of positive measure is here proven for the many-body magnetic Schr\"odinger equation. This property guarantees that if a solution of the Schr\"odinger equation vanishes on a set of positive measure, then it is identically zero. We explicitly consider potentials written as sums of either one-body or two-body functions, typical for Hamiltonians in many-body quantum mechanics. As a special case, we are able to treat atomic and molecular Hamiltonians. The unique-continuation property plays an important role in density-functional theories, which underpins its relevance in quantum chemistry. \newline

\noindent \textbf{Keywords.} unique-continuation property; Hohenberg--Kohn theorem; magnetic Schr\"odinger equation; Kato class; molecular Hamiltonian 
\end{abstract}

\maketitle

\section{Introduction}

Within the Schr\"odinger model for quantum systems of (interacting) electrons, in order to be able to describe interesting phenomena like the Zeeman effect, the quantum Hall effect, or the Hofstadter butterfly one has to include the effects of both an electric and a magnetic field. 
Hohenberg and Kohn showed for systems without magnetic fields that the one-body ground-state particle density determines the electric (scalar) potential up to a constant~\cite{Hohenberg1964}. Strictly speaking, the particle density determines at most one potential (modulo an additive constant) since some densities are not associated with any potential~\cite{ENGLISCH1983}. The above correspondence between densities and potentials constitutes the theoretical foundation on which Density-Functional Theory (DFT)---a ubiquitous tool in quantum chemistry and materials science~\cite{Burke2012,Jones2015}---is built.

In the presence of magnetic fields though, the approach of Hohenberg and Kohn to set up a universal density functional requires more than just the particle density due to the fact that an additional vector potential enters the system's Hamiltonian. Both the para\-magnetic current density and the total (physical) current density have been suggested as basic variables alongside the particle density \cite{Vignale1987,Vignale1988,Diener} and the resulting framework is called Current-Density-Functional Theory (CDFT). 
For the theory that uses the para\-magnetic current density, counterexamples to a Hohenberg--Kohn theorem are known~\cite{Capelle2002,Laestadius2014}, although a weaker version still holds. Note that even for degenerate systems the weaker version is enough to define a universal paramagnetic current-density functional~\cite{LaestadiusTellgren2018}. 
Diener has presented an argument~\cite{Diener} for establishing a full Hohenberg--Kohn theorem using the total current density. However, as first noted in Ref.~\cite{Tellgren2012}, the argument is at best incomplete.

For more detailed accounts on the existence of generalized Hohenberg--Kohn theorems within CDFT see Refs.~\cite{Tellgren2012,Laestadius2014,Tellgren2018}, and for related and positive results within the Maxwell--Schr\"odinger theory and quantum-electrodynamical DFT see Refs.~\cite{Ruggenthaler2015,TellgrenSolo2018,Garrigue2019}. An interesting and recent development is also given in Ref.~\cite{Garrigue2019b} where the existence of generalized Hohenberg--Kohn theorems is further explored. 
A different route, where a Hohenberg--Kohn result comes for free by virtue of the convex-analytic properties of a \emph{regularized} energy functional was taken in Refs.~\cite{Kvaal2014,KSpaper2018}. It was specifically implemented for CDFT in Ref.~\cite{MY-CDFTpaper2019} and can even be used to prove convergence of the associated Kohn--Sham iteration scheme~\cite{penz2019guaranteed}.

The current work arises as a natural ingredient for proving a generalized Hohenberg--Kohn theorem in total (physical) CDFT. It addresses the property that a solution of the magnetic Schr\"odinger equation cannot vanish on a set of positive measure, a property called unique continuation, see Definition~\ref{def:ucp}. Unique continuation is also a fundamental property for solutions of the magnetic Schr\"odinger equation in its own right and has been well-studied~\cite{H,BKRS,Wolff1992,Kurata1,Kurata2,Regbaoui}. 
In the context of CDFT the issue was first raised in Ref.~\cite{Laestadius2014}. 
As far as DFT and CDFT are concerned, it is useful to have the assumptions guaranteeing the unique-continuation property as particle-number independent as possible (at least avoid increasing itegrability constraints with increasing $N$), which is a
 difficult task. In the present work we obtain results that are adapted to the many-body Schr\"odinger equation 
and that furthermore include vector potentials, building on the results of Kurata \cite{Kurata2} and Regbaoui \cite{Regbaoui}. This means that the specific structure of the potentials is beneficially taken into account. The main results, Theorem~\ref{thm:5} and Corollary~\ref{cor:atoms}, that include the singular Coulomb potentials of atoms and molecules as a special case (Corollary~\ref{cor:mol-case}), are formulated in terms of the Kato class $K_\mathrm{loc}^n$ and its generalization $K_\mathrm{loc}^{n,\delta}$, with $n=3$ and $n=6$ (Definition~\ref{kato}).
Although we cannot answer the question of the existence of a generalized Hohenberg--Kohn theorem for the total current in CDFT, we exemplify the use of the unique-continuation property in a limited special case (Corollary~\ref{cor:cdft}).

\section{Unique-continuation property and the Hohenberg--Kohn theorem} 
In the most simple setting of only one particle and without vector
potential, it is known that the (unique) ground state $\psi$ in
$H^1(\mathbb{R}^3)$ can be chosen to be strictly positive, see Theorem 11.8 in Lieb--Loss~\cite{LiebLoss}. 
This means that we can set $\rho^{1/2}=\psi>0$ and the following relation to the scalar potential $v$ must hold (from the Schr\"odinger equation, here written as $(-\Delta + v)\psi = e \psi$)
\begin{equation} \label{eq:vin}
v(x)= e +  \frac{\Delta \rho(x)^{1/2}}{\rho(x)^{1/2}},\quad  x\in\mathbb R^3,
\end{equation}
where $\Delta$ denotes the Laplacian and $e$ is the ground-state energy. Conversely, given a particle density $\rho$ we can ask if a potential $v$ exists 
such that the given $\rho$ is the ground-state density of that potential. 
For the one-particle case, this problem has been studied by Englisch and Englisch~\cite{ENGLISCH1983} and was answered in the negative even for well-behaved densities ($N$-representable densities). Corollary~3 in Ref.~\cite{Laestadius2014b} (including the additional constraint $\Delta \rho^{1/2} \leq C \rho^{1/2}$ and $\rho^{-1}\in L_\mathrm{loc}^1$ besides $N$-representability) provides sufficient conditions for one-particle $v$-representability, i.e., $v$ can be computed from $\rho$ as given in Eq.~\eqref{eq:vin} and $\rho$ is the ground-state density of that $v$.

Returning to the general $N$-electron case without magnetic field, we first recall the Hohenberg--Kohn theorem: Given two systems, if $\rho_1 = \rho_2$ then $v_1 = v_2 \,+$ constant, where $\rho_k$, $k=1,2$, is the ground-state particle density of the corresponding system defined by the potential $v_k$. 
The proof of this result relies on the fact that if $\psi$ is a ground state of both systems, then 
$\sum_{k=1}^N(v_1(x_k) - v_2(x_k))\psi = \text{constant} \times \psi$. If $\psi$ does not vanish on a set of positive (Lebesgue) measure we have $v_1 =v_2\,+$ constant almost everywhere. 
The proof can then be completed by means of the variational principle, using the Hohenberg--Kohn argument by reductio ad absurdum~\cite{Hohenberg1964}. 
(Note that a strict inequality in the variational principle is not needed, see e.g.~\cite{Garrigue2018}, and that the results also hold for systems with degeneracy~\cite{ENGLISCH1983}.)

In this article we address the more general case of $N$ interacting, non-relativistic (spinless) particles
subjected to both a scalar and a vector potential. The fundamental question then is, whether any eigenfunction of the corresponding Hamiltonian
\begin{align}
\label{HN}
H_N = \sum_{j=1}^N \Big[(\i\nabla_j + A(x_j))^2 + v(x_j) +  \sum_{l < j} u(x_j,x_l) \Big]
\end{align}
can be zero on a set of positive measure without being identically zero. This is a problem of {\it unique continuation}.

\begin{definition} \label{def:ucp}
We say that the Schr\"odinger equation $H_N \psi = e\psi$ has the unique-continuation property (UCP) from sets of positive (Lebesgue) measure if 
a solution that satisfies $\psi=0$ on a set of positive measure is identically zero.  
Furthermore, the Schr\"odinger equation is said to have  
the strong UCP if whenever $\psi$ vanishes to infinite order at some point $x_0$, i.e., for all $m>0$
\begin{equation*}
\int_{|x-x_0|\leq r} |\psi(x)|^2 \d x = \mathcal O(r^m) \quad (r\to 0),
\end{equation*}
then $\psi$ is identically zero. Additionally, if $\psi=0$ on a non-empty 
open set implies that $\psi$ is identically zero, then the
Schr\"odinger equation has the weak UCP. 
\end{definition}
\begin{remark}
The strong UCP implies the weak UCP. 
The UCP from sets of positive measure allows us to conclude $\psi\neq 0$ almost everywhere for any eigenfunction of $H_N$. 
\end{remark}

There exists a considerable amount of literature that treats the UCP for differential inequality 
$|\Delta \psi| \leq |\xi_1| |\nabla \psi| + |\xi_2| |\psi|$~\cite{H,BKRS,Wolff1992,Kurata1,Kurata2,Regbaoui}. In particular if
$\xi_1 \in L_\mathrm{loc}^n(\mathbb R^n)$ and $\xi_2 \in
L_\mathrm{loc}^{n/2}(\mathbb R^n)$, the corresponding differential
equation $\Delta\psi=  \xi_1\cdot \nabla \psi + \xi_2 \,\psi $ has the UCP from sets of positive measure~\cite{Regbaoui}. 
Note that such $L^p_\mathrm{loc}$ constraints become more restrictive with
increasing particle number, since the dimension of the
configuration space $n$ enters in the conditions.
Directly applied to $H_N \psi = e\psi$ this means that if a solution $\psi$ in the Sobolev
space $H_\mathrm{loc}^{2N/(N+2)}(\mathbb R{^{3N}})$ vanishes on a set of positive measure,   
\begin{align*}
\sum_{j=1}^N \Big[v(x_j) + |A(x_j)|^2 &+\i (\nabla_j\cdot A(x_j)) \\
&+\sum_{ l < j} 	u(x_j,x_l)\Big] \in L_\mathrm{loc}^{3N/2}(\mathbb R^{3N}),
\end{align*}
and each component of $ A$ belongs to $L_\mathrm{loc}^{3N}(\mathbb
R^{3N})$, then $\psi$ is identically zero.

Such results are used by Lammert \cite{Lammert2018}, particularly in his Theorem~5.1, to give a mathematically precise proof of the Hohenberg--Kohn
theorem~\cite{Hohenberg1964} in DFT including the UCP as remarked by Lieb~\cite{Lieb1983}. Yet he does not consider magnetic fields and the constraints
are very susceptible to the particle-number. 
A recent effort by Garrigue \cite{Garrigue2018} removed the dependence on particle
numbers for the constraints on the scalar potential by exploiting their
specific shape in the context of many-body (molecular) Hamiltonians.
Ref.~\cite{Garrigue2018} also contains a rigorous proof of the Hohenberg--Kohn theorem including all the mathematical details for potentials $v\in L_\mathrm{loc}^p(\mathbb R^3)$, $p>2$.

\section{Prerequisites}
Let the Hamiltonian $H_N$ be as in \eqref{HN}. The Schr\"odinger
equation is then given by $H_N\psi = e \psi$. 
We write $H_N = T_A + V + U$, 
where $T_A= \sum_{j=1}^N (\i\nabla_j + A(x_j))^2$, $\nabla_j =
(\frac{\partial}{\partial x_j^1},\frac{\partial}{\partial
	x_j^2},\frac{\partial}{\partial x_j^3})$ and
$x_j=(x_j^1,x_j^2,x_j^3)\in \mathbb{R}^3$ are the coordinates of the
$j$:th electron. Here we use $(\i\nabla_j +  A(x_j) )^2$ in $T_A$ instead
of $(-\i\nabla_j +  A(x_j))^2$ in order to follow the notation in Kurata \cite{Kurata2}. We use a slight variation of atomic units $\hbar = 2m_e = 1$ and $q_e=-1$, such that the Laplace operator appears without a factor $1/2$.

The electric potential $V$ is a one-body potential given by 
$V(x)= \sum_{j=1}^N v(x_j)$ with $v:\mathbb{R}^3\rightarrow \mathbb{R}$. 
The two-particle interaction $U$ between the electrons is modeled by $U(x)= \sum_{1\leq j<l \leq N} u(x_j,x_l)$, for some non-negative function $u$ on $\mathbb R^3\times \mathbb R^3$. We set $W=V+U$. Furthermore,  
$A:\mathbb{R}^3\rightarrow\mathbb{R}^3$ denotes the vector potential,
from which the magnetic field is obtained by $B=\nabla \times A$. 
With the notation $\underline A(x) = (A(x_j))_{j=1}^N$, the Schr\"odinger
equation is rewritten as 
\begin{equation}
\label{SE2}
-\Delta\psi + 2\i \underline A \cdot \nabla \psi + (W_A - e)\psi=0, 
\end{equation}
where $W_A = W + |\underline A|^2 + \i(\nabla \cdot\underline A)$.

A function $f\in L_{\mathrm{loc}}^2(\mathbb R^n)$ belongs to the Sobolev
space $H_{\mathrm{loc}}^k(\mathbb R^n)$ if $f$ has weak derivatives up to order $k$ that belong to $L_{\mathrm{loc}}^2(\mathbb R^n)$. 
Let the set	of infinitely differentiable functions with compact support on $\mathbb R^{3N}$ be denoted by $C_0^\infty(\mathbb R^{3N})$. We say that 
$\psi\in H_\mathrm{loc}^1(\mathbb R^{3N})$ is a solution of~\eqref{SE2} in the weak sense, which will be our standard notion for solutions from here on,  
if for all $\varphi \in C_0^\infty(\mathbb R^{3N})$

\begin{align}
\int_{\mathbb R^{3N}} \nabla \psi \cdot \nabla \overline{\varphi}\,\d x &+ 2\i \int_{\mathbb R^{3N}}  \underline A \cdot (\nabla \psi) \overline{\varphi}\,\d x  \nonumber  \\
&+ \int_{\mathbb R^{3N}} (W_A - e)\psi \overline{\varphi}\, \d x=0.  \label{SE3}
\end{align}
The present work takes off from the following result: 
\begin{theorem}[Theorem~1.2 in Regbaoui \cite{Regbaoui}] Let $N\geq 1$.
	Assume that $W_A\in L_\mathrm{loc}^{3N/2}(\mathbb R^{3N})$ and each component of
	$ A$ is an element of $L_\mathrm{loc}^{3N}(\mathbb R^{3N})$. Then the
	Schr\"odinger equation has the UCP from sets of positive
	measure, i.e., if a solution $\psi\in H_\mathrm{loc}^{2N/(N+2)}(\mathbb R^{3N})$
	vanishes on a set of positive measure then it is identically 
	zero. \label{Thm:Reg}
\end{theorem}
\begin{remark}
	See also Theorem~1.1 in Regbaoui \cite{Regbaoui} for the strong UCP and Wolff \cite{Wolff1992} for the weak UCP. 
\end{remark}

If one employs Theorem~\ref{Thm:Reg} with $N=1$, since there is no two-particle interaction it suffices to assume that  $v,|A|^2$ and $\nabla\cdot A$ are elements of 
 $L_\mathrm{loc}^{3/2}(\mathbb R^3)$ to obtain the UCP from sets of
positive measure. 
With increasing particle number, however, the assumptions on the
potentials $v$, $u$, and $A$ are such that they rule out most types of
singularities. 
On the other hand, the particle-number dependence that enters in $H_\mathrm{loc}^{2N/(N+2)}$ fulfills $2N/(N+2)\leq 2$ for all $N$. 
Following Kurata~\cite{Kurata2} the $L_\mathrm{loc}^p$ constraints, with $p$
proportional to $N$, can be avoided.

\begin{definition}\label{kato}
	A function $f\in L_\mathrm{loc}^1(\mathbb R^n)$ belongs to the Kato class $K_{\mathrm{loc}}^n$, $n\neq 2$, if for every $R>0$, $\lim_{r\rightarrow 0+} \eta^K(r;f)=0$, where
	\[  
	\eta^K(r;f) = \sup_{|x| \leq R} \int_{B_{r}(x)}\frac{|f(y)|}{|x-y|^{n-2}}\d y.
	\]
	Furthermore, $f\in K_\mathrm{loc}^{n,\delta} \subset K_\mathrm{loc}^n$, $\delta>0$, if for every $R>0$
	\[
	\lim_{r\to 0+} \sup_{|x| \leq R} \int_{B_{r}(x)}\frac{|f(y)|}{|x-y|^{n-2+ \delta}}\d y=0.
	\]
\end{definition}

We write $f= f_+-f_-$, where $f_-$ ($f_+$) is the negative (positive) part of $f$ given by $f_-(x) = \max(-f(x),0)$ ($f_+(x) = \max(f(x),0)$). 
Let $y\in \mathbb R^{3N}$ be fixed and for
$x=(x_1,\dots,x_N)$ $\in{\mathbb R}^{3N}$ (cf.~the notation in Kurata \cite{Kurata2})
\begin{align*}
a(x) &=\vert \underline A(x) \vert^2, \\ 
b_{y}(x) &= |x-y|^2\sum_{j=1}^N |B(x_j)|^2,\\
Q_{y}(x) &= (2W + (x-y) \cdot \nabla W)_{-}\,.
\end{align*}
With the notation above, we formulate 
\begin{assumption}\label{A:1}
	Suppose
	\begin{align*}
	    &\nabla\cdot A \in L_\mathrm{loc}^2(\mathbb R^3), \\
	    & A^j\in L_\mathrm{loc}^4(\mathbb R^3)\quad \text{($A^j$ component of $A$)} , \\
	    & a, b_{y},W, Q_{y} \in K_{\mathrm{loc}}^{3N} \quad \text{(for fixed $y\in \mathbb R^{3N}$)},
	\end{align*}
	and that for some $r_0>0$
	\begin{equation}
	\int_0^{r_0} \frac{\theta_y(r)}{r} \d r< \infty
	\label{theta}
	\end{equation}
	holds, where $\theta_y(r) = \eta^K(r;Q_{y}) + \eta^K(r; b_{y})^{1/2}$.
\end{assumption}
\begin{remark}\label{rmk:Kurata}
	Remark~1.2 in Kurata \cite{Kurata2} gives $b_{y},Q_{y}\in
	K_\mathrm{loc}^{3N,\delta}$, for some $\delta>0$, as a sufficient condition for
	$\eqref{theta}$ to hold.  	
\end{remark}
\begin{remark}
The main condition in Assumption~\ref{A:1} is with respect to the Kato class
$K_\mathrm{loc}^n$, $n=3N$ being the dimensionality of the underlying configuration
space. This condition is \emph{optimal} in the sense that the class
cannot be enlarged to \emph{smaller} orders than $n=3N$, or the UCP
will be lost. This follows from the inclusion $L_\mathrm{loc}^p \subset
K_\mathrm{loc}^n$ for all $p > n/2$ and a sharp counterexample provided in
Ref.~\cite{Koch-Tataru2007} for a potential in $L^p$, $p < n/2$. So if the order of
the Kato class would be any $m < n$ then it also includes $L_\mathrm{loc}^p$ 
with $m/2 < p < n/2$ and that is ruled out by the given counterexample.
\end{remark}

The following is obtained by adapting Corollary~1.1 in Kurata \cite{Kurata2} (denoted Lemma~\ref{Lemma} below). For the sake of simplicity, and since it is enough for our purposes here, we make the restrictions to real-valued $V$ and $U$. In the sequel we use the notation $\vert F \vert$ for the Frobenius norm (also called the Hilbert--Schmidt norm) of a matrix $(F_{j,l})_{j,l}$. 
\begin{theorem}\label{th:strongUCP}
	Suppose Assumption~\ref{A:1}. 
	If $\psi \in
	H_{\mathrm{loc}}^2 (\mathbb R^{3N})$ is a solution of \eqref{SE2}
	and vanishes to infinite order at $x_0 \in R^{3N}$, then $\psi$ is
	identically zero. Thus the Schr\"odinger equation has the strong UCP.\label{Thm:adaptedK}
\end{theorem}
\begin{lemma}[Corollary~1.1 in Kurata \cite{Kurata2}]\label{lemma:Kurata}
	Let $n\geq 3$, $x_0\in \mathbb R^n$ be fixed, $x=(x^1, \dots,x^n)\in\mathbb R^n$, $\tilde A =(\tilde{A}^1,\dots, \tilde{A}^n): \mathbb{R}^n\to \mathbb{R}^n$, $\tilde{W}: \mathbb{R}^n\to \mathbb{R}$, 
	$F = (F_{j,l})_{j,l=1}^n$ with $F_{j,l} = \partial \tilde A^j/\partial x^l - \partial \tilde A^l/\partial x^j$, and suppose that 
	\begin{align}
	&\tilde{A}^j\in  L_\mathrm{loc}^4(\mathbb R^n),\quad \nabla\cdot \tilde A\in L_\mathrm{loc}^2(\mathbb R^n), 
	\quad \vert \tilde A\vert^2 \in K_\mathrm{loc}^n, \nonumber \\
	&  (\vert x-x_0\vert \,\vert F\vert)^2\in K_{\mathrm{loc}}^n \label{Cond1},
	\end{align}
	and
	\begin{align}
	\tilde{W} \in K_\mathrm{loc}^n, \quad (2\tilde W + (x-x_0)\cdot \nabla \tilde W)_-\in K_{\mathrm{loc}}^n.\label{Cond2}
	\end{align}
	Furthermore assume, for some $r_0>0$, 
	\begin{align}
	\int_0^{r_0}  \Big[ \eta^K(r;(2\tilde W &+ (x-x_0) \cdot \nabla \tilde W)_{-}) \nonumber \\
	&+ \eta^K(r; (\vert x-x_0\vert \,\vert F\vert)^2  )^{1/2}\Big] \frac{\d r}{r} < \infty
	\label{thetatilde}
	\end{align}
	holds. Then if $\psi \in H_\mathrm{loc}^{2}(\mathbb R^n)$ satisfies 
	\begin{align} \label{eq:KurataSE}
	\Big(\sum_{j=1}^n \Big( \i\frac{\partial}{\partial x^j} + \tilde{A}^j(x)  \Big)^2 + \tilde{W}(x) \Big)\psi =0
	\end{align}
	and vanishes to infinite order at $x_0$, it follows that $\psi$ is identically zero. 
	\label{Lemma}
\end{lemma}
\begin{proof}[Proof of Theorem~\ref{Thm:adaptedK}]
	We will show that Assumption~\ref{A:1} directly fulfills all the conditions of Lemma~\ref{lemma:Kurata} that thus becomes applicable. 
	Let $n=3N$ and $\tilde W = W -e$. By Assumption~\ref{A:1}, \eqref{Cond2} is then fulfilled. The choice $\tilde A = \underline A$  
	implies $T_A= (\i\nabla + \tilde A)^2$ and \eqref{SE2} can be written as~\eqref{eq:KurataSE}.

	Each component of $A\in L_\mathrm{loc}^4(\mathbb R^3)$ yields $\tilde A^j \in L_\mathrm{loc}^4(\mathbb R^{3N})$ for $j=1,\dots,3N$. 	
	From $\nabla \cdot \tilde A = \sum_{k=1}^N \nabla_k\cdot A(x_k)$ and
	$\nabla \cdot A \in L_\mathrm{loc}^2(\mathbb R^3)$, we obtain 
	$\nabla \cdot \tilde A \in L_\mathrm{loc}^2(\mathbb R^{3N})$. Since
	$a = |\tilde A|^2$, it holds that $|\tilde A|^2\in K_\mathrm{loc}^{3N}$. Moreover, the matrix $F$ satisfies
	\begin{equation}
	(|x-x_0|\,|F|)^2 = |x-x_0|^2 \sum_{j,l=1}^{3N} |F_{j,l}|^2 = 2 b_{x_0}(x), 
	\label{Ftob}
	\end{equation}
	since $F$ contains $N$ repeated blocks of sub matrices of the form
	\[
	\begin{bmatrix}
	0 & -B_3(x_j) & B_2(x_j)\\
	B_3(x_j) & 0 & -B_1(x_j)\\
	- B_2(x_j) & B_1(x_j) & 0
	\end{bmatrix}.
	\]
	This establishes \eqref{Cond1}.
	
	From \eqref{theta}, $\tilde W = W -e$ and \eqref{Ftob}, we conclude that \eqref{thetatilde} holds. Lemma~\ref{Lemma} gives the strong UCP for \eqref{SE2} and the proof is complete. 
\end{proof}

\begin{remark}
	As stated in Remark~1.1 in Kurata~\cite{Kurata2}, Lemma~\ref{lemma:Kurata} and thus Theorem~\ref{Thm:adaptedK} also holds if in Assumption~\ref{A:1},
	$K_\mathrm{loc}^{3N}$ is replaced by $K_\mathrm{loc}^{3N} +
	F_\mathrm{loc}^p(\mathbb{R}^{3N})$, $1<p \leq 3N/2$. Here $F_\mathrm{loc}^p(\mathbb{R}^{3N})$ is the
	Fefferman--Phong class and in this case a solution must be an element of $H_\mathrm{loc}^2(\mathbb{R}^{3N})  \cap L_\mathrm{loc}^\infty(\mathbb{R}^{3N}) $, and there is an additional condition on $V_-$. 
\end{remark}

\section{Main Results}
Theorem~\ref{Thm:adaptedK} above establishes the strong UCP under Assumption~\ref{A:1}. If in addition the negative part of $v$ is locally $L^{3/2}(\mathbb R^3)$ summable we obtain the UCP from sets of positive measure:
\begin{theorem} \label{thm:5}
	Suppose Assumption~\ref{A:1}. 
	If in addition $v_-\in L_\mathrm{loc}^{3/2}(\mathbb{R}^3)$ and $\psi \in
	H_{\mathrm{loc}}^2(\mathbb{R}^{3N})$ solves \eqref{SE2} and vanishes on a set of
	positive measure, then $\psi$ is identically zero. Consequently, the Schr\"odinger equation has the UCP from sets of positive measure.\label{Thm:Main}
\end{theorem}
\begin{remark}
The requirement $u\geq 0$ can be relaxed if one assumes that $u(x_1,x_2) = u'(x_1-x_2)$ (see
Lemma~A.2 in Lammert \cite{Lammert2018}).
\end{remark}
If the strong UCP can be obtained under other assumptions than Assumption~\ref{A:1}, the following corollary can be used to obtain the UCP from sets of positive measure.
\begin{corollary}\label{cor:strong2measure}
	Suppose the strong UCP for the Schr\"odinger equation (not necessarily by means of Assumption~\ref{A:1}), then the constraint $v_- + |A|^2 + \i (\nabla\cdot A) \in L_\mathrm{loc}^{3/2}(\mathbb R^3)$ gives the UCP from sets of positive measure.  
\end{corollary}

Due to the particular form of the potentials, we can write
\[
W(x) = \sum_j v(x_j)+\frac 1 2 \sum_{j\neq l} u(x_j,x_l).
\]
Because $Q_y$ is defined as the negative part of the function $2W + (x-y)\cdot \nabla W$, we have 
 with the choice $\underline x_0=( x_0,\dots, x_0) \in \mathbb R^{3N}$, for fixed $ x_0 \in \mathbb R^3$,  
\begin{equation}\label{eq:bigQ}
0\leq Q_{x_0}(x) = Q_{\underline x_0}(x) \leq \sum_{j} q_{1;x_0}(x_j) + \sum_{j\neq l} q_{2;x_0}(x_j,x_l)
\end{equation}
for some functions $q_{1;x_0}$ and $q_{2;x_0}$. (See below the proof of Corollary~\ref{cor:mol-case}, where this decomposition is done for the choice of $W$ corresponding to the molecular case.) 
Furthermore, 
\[
b_{x_0}(x) = b_{\underline x_0}(x)= \sum_{j,l=1}^{N}  |x_j-x_0|^2 |B(x_l)|^2
\]
can be split as 
\[
b_{x_0}(x) = \sum_{j} b_{1;x_0}(x_j) + \sum_{j\neq l} b_{2;x_0}(x_j,x_l).
\]

We can now formulate our main result that includes $H_N$ modelling atoms and molecules, and where the exponents in the integrability constraints are independent of the particle number $N$.

\begin{corollary}\label{cor:atoms}
For $N\geq 2$ and $ x_0 \in \mathbb R^3$ fixed, suppose
 \begin{equation} \label{eq:magA}
     \begin{split}
         & \nabla\cdot A \in L_\mathrm{loc}^2(\mathbb R^3),\quad A^j \in L_\mathrm{loc}^4(\mathbb R^3), \quad \vert A\vert^2\in K_\mathrm{loc}^3, \\
         &  b_{1;x_0}\in K_\mathrm{loc}^{3,\delta}, \quad b_{2;x_0}\in K_\mathrm{loc}^{6,\delta}.
     \end{split}
 \end{equation}
Further, let $v_-\in L_\mathrm{loc}^{3/2}(\mathbb R^3)$, $v \in K_\mathrm{loc}^3$, and $u\in K_\mathrm{loc}^6$, as well as $Q_{x_0}$ satisfying \eqref{eq:bigQ} with 
$q_{1;x_0}\in K_\mathrm{loc}^{3,\delta}$ and $q_{2;x_0}\in K_\mathrm{loc}^{6,\delta}$.
Then the Schr\"odinger equation \eqref{SE2} has the UCP from sets of positive measure. 
\end{corollary}
In particular, the magnetic Schr\"odinger equation has the UCP from sets of positive measure for $H_N$ modelling atoms and molecules in magnetic fields if just Eq.~\eqref{eq:magA} is fulfilled.
\begin{corollary} \label{cor:mol-case}
    For $N\geq 2$, suppose  the magnetic field is such that Eq.~\eqref{eq:magA} holds. Then with
\[
v(x_1)= -\sum_{j=1}^{M_\mathrm{nuc}}\frac{Z_j}{\vert x_{\mathrm{nuc};j} - x_1 \vert},\quad u(x_1,x_2) = \frac{1}{\vert x_1-x_2\vert},
\]
where $x_{\mathrm{nuc};j}\in\mathbb R^3$ and $Z_j>0$ are the positions and charges of the $M_\mathrm{nuc}$ nuclei, respectively, the UCP from sets of positive measure holds for the Schr\"odinger equation. 
\end{corollary}

\section{Application to CDFT}

In the presence of a magnetic field, no equivalence of a (general) Hohenberg--Kohn result exists at present~\cite{Tellgren2012,Laestadius2014}. However, we shall now address how the UCP from sets of positive measure for the magnetic Schr\"odinger equation plays an important role in the argument for restricted 
Hohenberg--Kohn theorems in CDFT and the non-universal variant magnetic-field density-functional theory (BDFT) of Grayce and Harris~\cite{GRAYCE}. Given a wave function $\psi$, define the particle density and the paramagnetic current density according to
\begin{align*}
\rho_\psi(x) &= N \int_{\mathbb{R}^{3(N-1)}} |\psi(x,x_2,\dots,x_N)|^2\, \mathrm{d} x_2\dots \mathrm{d} x_N,\\
j_\psi^p(x) &= N\,\text {Im} \int_{\mathbb{R}^{3(N-1)}} \overline{\psi}(x,x_2,\dots,x_N) \\
&\quad \quad \quad \times \nabla_{x}\psi(x,x_2,\dots,x_N)  \, \mathrm{d}x_2\dots \mathrm{d} x_N.
\end{align*}
For a vector potential $A$ we may compute the total current density by the sum
$j = j_\psi^p + \rho_\psi A$. Now, fix the particle number $N$ as well as the two-particle interaction $u$ (e.g.\ $u(x_1,x_2) =|x_1-x_2|^{-1}$) and write $H_N =H(v,A)$. 
If $\psi$ is a ground state for some $v$ and $A$, i.e., $H(v,A)\psi =e\psi$, where $e$ is the ground-state energy, 
then $\rho_{\psi}$, $j_{\psi}^p$, and $j= j_{\psi}^p + \rho_{\psi}A$ are called ground-state densities of $H(v,A)$. 
Whether the ground-state particle density $\rho$ and the total current density $j$ determine $v$ and $A$ (up to a gauge transformation) is 
still an open question in the general case~\cite{Tellgren2012,Laestadius2014}. (For the ground-state density pair $(\rho,j^p)$ it is well-known that this density pair does not determine the potentials $v$ and $A$~\cite{Capelle2002}.)

Now, assume that two systems with Hamiltonians $H(v_1,A_1)$ and $H(v_2,A_2)$ have the same ground-state particle density (i.e.\ $\rho_1=\rho_2=\rho$) and $\nabla\times A_1 =\nabla\times A_2=B$. Suppose that $v_k$, $A_k$ for $k=1,2$, and $B$ fulfill Assumption~\ref{A:1} and the requirements given in Theorem~\ref{thm:5}. Since there exists a function $f$ such that $A_1=A_2 - \nabla f$, the 
variational principle yields
\begin{align*}
    e_1 \leq \langle \psi_2, H(v_1,A_1) \psi_2 \rangle \leq e_2 + \int_{\mathbb R^2} (v_1-v_2) \rho \mathrm{d} x,
\end{align*}
where $\psi_2$ is the ground state of $H(v_2,A_2-\nabla f)$.
Switching the indices, we find that $$e_1-e_2 = \int_{\mathbb R^2} (v_1-v_2) \rho \mathrm{d} x.$$ Consequently, $\psi_2$ is a ground state of both $H(v_1,A_1)$ and $H(v_2,A_2 - \nabla f)$, which leads to 
\begin{align*}
&\left[H(v_1,A_1) -H(v_2,A_2-\nabla f) \right]\psi_2 \\
&\quad = \sum_{j=1}^N (v_1(x_j)- v_2(x_j))\psi_2 \\
&\quad = (V_1-V_2)\psi_2 = (e_1-e_2)\psi_2.
\end{align*}
However, Theorem~\ref{thm:5} allows us to conclude $\psi_2\neq 0$ and it follows $v_1=v_2 +$constant. 

\begin{theorem} 
	Assume $\nabla\times A_1 =\nabla\times A_2=B$ and that $v_k$, $A_k$ for $k=1,2$, and $B$ fulfill Assumption~\ref{A:1} and take the requirements of Theorem~\ref{thm:5} for $H(v_1,A_1)$ and $H(v_2,A_2)$ to hold. If the ground-state particle densities satisfy $\rho_1=\rho_2$, then $v_1=v_2 + C$ almost everywhere for some constant $C$.
	\label{BDFT}
\end{theorem}
\begin{remark} Theorem \ref{BDFT} is the Hohenberg--Kohn theorem for BDFT, first established by Grayce and Harris~\cite{GRAYCE} but missing the UCP argument (see also~\cite{Laestadius2014}).
\end{remark}
Theorem \ref{BDFT} can be used to obtain 
\begin{corollary} \label{cor:cdft}
	Assume Assumption~\ref{A:1} and the requirements of Theorem~\ref{thm:5} for $H(v_1,A_1)$ and $H(v_2,A_2)$ and that the ground-state densities fulfill $\rho = \rho_1=\rho_2$, $j= j_1=j_2$. For systems with $j^p=0$, it follows $B_1=B_2$ (even $A_1=A_2$ holds) and $v_1=v_2+C$ for some constant $C$.
\end{corollary}
\begin{proof}
For systems with $j^p = 0$, $j_1=j_2$ implies $$\rho A_1=\rho A_2,$$ 
since $\rho_1=\rho_2=\rho$. Theorem~\ref{thm:5} gives $\rho>0$ almost everywhere and we may conclude $A_1=A_2$. Theorem~\ref{BDFT} now gives $v_1=v_2 + C$ for some constant $C$.
\end{proof}

\section{Proofs of the main results}
\label{Sec:Proof}
\begin{proof}[Proof of Theorem~\ref{Thm:Main}] In the sequel let $D=3N$. By Assumption~\ref{A:1}, the strong UCP holds for \eqref{SE2} by Theorem~\ref{Thm:adaptedK}. Next, we follow the proof of Theorem~1.2 given after Lemma~3.3 in Regbaoui \cite{Regbaoui} and Lemma~A.2 in Lammert \cite{Lammert2018}. (Lemma~A.2 corresponds to setting $A=0$ here, and moreover we exploit $u\geq 0$ instead of the assumption $u(x_1,x_2)=u'(x_1-x_2)$.) We start by showing the following
	inverse Poincar\'e inequality for solutions of the Schr\"odinger
	equation: 
	
	For an arbitrary point $x_0 =(x_{0;j})_{j=1}^N \in\mathbb R^D$ and $r\leq r_0$
	\begin{equation}
	\label{Regbaoui1}
	\int_{B_{r}(x_0)} |\nabla \psi(x)|^2 \d x \leq \frac{C}{r^2}\int_{B_{2r}(x_0)} |\psi(x)|^2 \d x .
	\end{equation}
	Here $C$ is a positive constant that depends on $r_0>0$, $v$, and $A$ (but is independent of $u\geq 0$). 
	
	Choose $h \in C_0^\infty(B_{2r}(x_0))$ that satisfies $h(x) = 1$ if $|x-x_0|\leq r$, $h\leq 1$ for $|x-x_0|\leq 2r$, and $|\nabla h(x)| \leq 2r^{-1}$. In the Schr\"odinger equation \eqref{SE3}, 
	we choose $\varphi = h^2 \psi$ and move all terms except one to the right hand side so that 
	\begin{align}
	&\int_{\mathbb R^{D}} |h\nabla \psi|^2 \d x = -2\int_{\mathbb R^{D}} h (\nabla \psi) \cdot \overline{\psi} \nabla h \, \d x \nonumber \\
	&\quad - 2\i \int_{\mathbb R^{D}} \underline A \cdot (\nabla\psi) h^2\overline{\psi} \, \d x   +  \int_{\mathbb R^{D}}(e -W_A )|h\psi|^2 \d x.
	\label{mainEq}
	\end{align}
	We now bound each of the terms of the right hand side in \eqref{mainEq}. 
	
	It is immediate that the first term is less or equal to $2 \Vert h \nabla \psi \Vert_2 \Vert \psi \nabla h \Vert_2$. Using the inequality $2ab \leq a^2/6 + 6 b^2$, we obtain an upper bound 
	\begin{equation}
	\frac 1 6 \int_{\mathbb R^{D}} |h\nabla \psi|^2 \d x+ 6 \Vert \psi \nabla h \Vert_2^2. 
	\label{eq:I1}
	\end{equation}

	To continue, let $I_1 = - 2\i \int_{\mathbb R^{D}} \underline A
	\cdot (\nabla\psi) h^2\overline{\psi} \, \d x$. The
	Cauchy--Schwarz inequality together with $2ab \leq a^2/6 + 6 b^2$ yield   
	\begin{align}
	I_1 
	\leq  \frac 1 6  \int_{\mathbb R^{D}} |h\nabla \psi|^2 \d x + 6 \int_{\mathbb R^D} |\underline A|^2 |h \psi|^2 \d x.
	\label{eq:I2}
	\end{align}

	For the last term of the right hand side in \eqref{mainEq}, we use the definition of $W_A$ and write 
	$e-W_A  = e- W - |\underline A|^2 - \i (\nabla \cdot \underline A)$. 
	Thus
	\begin{align*}
	\int_{\mathbb R^{D}}(e -W_A )|h\psi|^2 \d x &= \int_{\mathbb R^D} (e- W - |\underline A|^2) |h\psi|^2 \d x \\
	&\quad -\i \int_{\mathbb R^D} (\nabla\cdot \underline A) |h\psi|^2 \d x
	\end{align*}
	and it follows from $W = V_+ + U_+ - V_-\geq - V_-$ that
	\begin{align}
	\int_{\mathbb R^{D}}(e -W_A )|h\psi|^2 \d x &\leq \int_{\mathbb R^D}( V_- + |e|) |h\psi|^2 \d x  \nonumber \\
	&\quad + \int_{\mathbb R^D} |(\nabla\cdot \underline A)| |h\psi|^2 \d x.
	\label{WAbound}
	\end{align}

	Define $\Theta = V_- + |e| + 6|\underline A|^2 + |\nabla \cdot \underline A|$, from \eqref{eq:I1}, \eqref{eq:I2}, and \eqref{WAbound} we obtain an upper bound for the right hand side of \eqref{mainEq} 			given by   
	\begin{align}
	\frac 1 3 \int_{\mathbb R^{D}} |h\nabla \psi|^2 \d x +  6 \Vert \psi \nabla h \Vert_2^2 + \int_{\mathbb R^D}  \Theta |h\psi|^2 \d x.
	\label{eq:I123}
	\end{align}
	With the notation $ \Theta_1 =  v_-+ N^{-1}|e| +6|A|^2 +|\nabla \cdot A| $,
	the inequality 
	$
	\Theta(x) \leq  \sum_{j=1}^N \Theta_1(x_j)
	$ 
	holds. Furthermore, we have
	\[
	\int_{\mathbb R^D}  \Theta(x) |h\psi|^2 \d x \leq \sum_{j=1}^N\int_{\mathbb R^D}  \Theta_1(x_j) |h\psi|^2 \d x =I_2,
	\]
	where the last equality defines $I_2$.

By assumption $v_-\in L_\mathrm{loc}^{3/2}(\mathbb R^3) $, 
$|A|^2\in L_\mathrm{loc}^2(\mathbb R^3)$, and $\nabla \cdot A \in
	L_\mathrm{loc}^2(\mathbb R^3)$, and it follows that $\Theta_1 \in L_\mathrm{loc}^{3/2}(\mathbb R^3)$.	
	To bound the term $I_2$ from above, 
	we closely follow Lammert \cite{Lammert2018} and define $\tilde \rho(x_1) = N\int_{\mathbb R^{3(N-1)}}  |h\psi|^2 \d x_2\cdots \d x_N$. For $M>0$ we let 
	$M' = \Vert \Theta_1\chi_{B_{2r}(x_{0;1}) } \chi_{\{ \Theta_1 \geq M \}} \Vert_{3/2}$, where the characteristic function of a set $X$ is denoted $\chi_X$.
	H\"older's inequality gives 
	\begin{align*}
	I_2 &= \int_{\{\Theta_1< M\}} \Theta_1 \tilde \rho\,  \d x_1 + \int_{\{\Theta_1 \geq M\}}  \Theta_1 \tilde \rho \,\d x_1 \\
	&\leq M \Vert h\psi\Vert_2^2 +	M' \Vert \tilde \rho \Vert_3,
	\end{align*} 
	and by a Sobolev inequality $\Vert \tilde \rho\Vert_3\leq C \Vert \nabla (\tilde \rho^{1/2} ) \Vert_2^2$. A direct computation of 
	$ \nabla (\tilde \rho^{1/2} )$, using the definition of $\tilde \rho$, shows that $\Vert \nabla (\tilde \rho^{1/2} ) \Vert_2^2 \leq \int_{\mathbb R^D} |\nabla (h\psi)|^2 \d x$ (see also the original argument of Lieb~\cite{Lieb1983}, Theorem~1.1). From $|a+b|^2 \leq 2|a|^2 + 2|b|^2$, we get
	\begin{equation*}
	\Vert \nabla (\tilde \rho^{1/2} ) \Vert_2^2  \leq  2 \int_{\mathbb R^{D}} |h\nabla\psi|^2 \d x + 2 \Vert \psi \nabla h \Vert_2^2.
	\end{equation*}
	We choose $M>0$ such that $2CM' \leq 1/6$ and then one has
	\begin{equation}
	\label{eq:I3}
	I_2 \leq\frac 1 6  \int_{\mathbb R^{D}} |h\nabla\psi|^2 \d x + \frac 1 6 \Vert \psi \nabla h \Vert_2^2 + M \Vert h\psi\Vert_2^2.
	\end{equation}

	Returning to \eqref{eq:I123}, we set $C = [74 + 3M r_0^2]/3$ and use \eqref{eq:I3} and $|\nabla h|\leq 2/r$ to conclude for $r\leq r_0$
	\begin{align*}
	\frac 1 2 \int_{\mathbb R^{D}} |h\nabla \psi|^2  \d x&\leq  \frac {37} 6 \Vert \psi \nabla h \Vert_2^2 + M \Vert h\psi \Vert_2^2 \\
	&\leq \frac {C}  {r^2} \int_{B_{2r}(x_0)} |\psi|^2 \d x.
	\end{align*}
	Hence \eqref{Regbaoui1} holds.

Suppose $\psi\in H_\mathrm{loc}^2$ vanishes on a set $E$ of positive measure.
	Almost every point  of $E$ is a density point. Let $x_0$ be such a density point and let $B_r = B_{r}(x_0)$. Given
	$\varepsilon>0$ there is an $r_0=r_0(\varepsilon)$ so that (cf.\ (3.11) in Regbaoui \cite{Regbaoui})
	\begin{equation}
	\frac{|{E}\cap B_{r}|}{|B_{r}|} \geq 1 - \varepsilon, \quad\frac{|{E}^c\cap B_{r}|}{|B_{r}|}\leq\varepsilon, \quad \text{for}\quad r\leq r_0.
	\label{EcE}
	\end{equation}
	Lemma 3.3 in Regbaoui \cite{Regbaoui} (or Lemma~3.4 in Ladyzen\-skaya--Ural'tzeva \cite{Ladyzenskaya1968}) gives
	\begin{equation}
	\int_{B_r \cap E^c}|\psi|^2 \d x \leq C \frac{r^D}{|E|} |B_r \cap E^c|^{1/D} \int_{B_r} |\nabla (\psi^2)|\d x 
	\label{lemma33}
	\end{equation}
	for some constant $C$.
	Applying the Cauchy--Schwarz inequality to the right hand side of \eqref{lemma33}, we obtain
	\[
	\int_{B_r }|\psi|^2 \d x\leq  C \frac{r^{2D}}{|E|^2} |B_r \cap E^c|^{2/D} \int_{B_r} |\nabla \psi|^2 \d x
	\]
	for some new constant $C$. Since $|E|\geq |E \cap B_r|$,
	\eqref{EcE}, and the inverse Poincar\'e inequality
	\eqref{Regbaoui1} allow us to conclude that
	\begin{align}
	\int_{B_r }|\psi|^2 \d x &\leq  C \frac{\varepsilon^{2/D}}{(1-\varepsilon)^2} r^2 \int_{B_r} |\nabla \psi|^2 \d x \nonumber \\
	&\leq C' \frac{\varepsilon^{2/D}}{(1-\varepsilon)^2}  \int_{B_{2r}} |\psi|^2 \d x.
	\label{rto2r}
	\end{align}

Introduce the function $f(r)=\int_{B_r} |\psi|^2 \d x$, fix an integer $n$ and choose $\varepsilon>0$ so that
	$C'\varepsilon^{2/D}/(1-\varepsilon)^2 =2^{-n}$. Then \eqref{rto2r} can be written $f(r)\leq 2^{-n}f(2r)$. By iteration 
	\[
	f(r') \leq 2^{-kn} f(2^kr'),\quad r'\leq 2^{1-k}r_0
	\]
	holds. For fixed $r$ and $k$ chosen such that $2^{-k}r_0 \leq r \leq 2^{1-k}r_0$, it follows that
	\[
	f(r)\leq 2^{-kn} f(2r_0) \leq \left(\frac{r}{r_0}\right)^n f(2r_0),
	\]
	where $r_0$ depends on $n$. Consequently $f$ vanishes to infinite order, i.e., for all $m$ there is an $r_0(m)$ such that
	\[
	f(r)=\int_{B_r}|\psi|^2 \d x \leq C_mr^m,\qquad r\leq r_0(m).
	\]
	That $\psi = 0$ follows now by the strong UCP given by Theorem~\ref{th:strongUCP}.
\end{proof}
\begin{proof}[Proof of Corollary~\ref{cor:strong2measure}]
	This is a consequence of the proof of Theorem~\ref{Thm:Main}, since $\Theta_1$, by assumption, is an element of $L_\mathrm{loc}^{3/2}(\mathbb R^3)$.
\end{proof}
\begin{proof}[Proof of Corollary~\ref{cor:atoms}]
We first demonstrate that the conditions of Corollary~\ref{cor:atoms} fulfills Assumption~\ref{A:1}. 
Due to the particular form of the potentials, we make use of the following: Let $f_1 \in K_\mathrm{loc}^{3,\delta}$ and $f_2 \in K_\mathrm{loc}^{6,\delta}$. Then both $\sum_{k=1}^N f_1(x_k)$ and $\sum_{k\neq l} f_2(x_k,x_l)$ are elements of $K_\mathrm{loc}^{3N,\delta}$. 
Similar statements for $K^n$ can be found in Simon \cite{Simon1982} (Example~F) and Aizenman--Simon \cite{Aizenman} (Theorem~1.4). 
We prove our claim by direct computations. Define $I_1^\delta$ and $I_2^\delta$ according to
\begin{align*}
I_1^\delta(x)&= \sum_{j=1}^N\int_{B_{r}(0)}\frac{\vert f_1(y_j+x_j) \vert  }{(y_1^2+\cdots + y_N^2)^{\frac{3N-2+\delta}{2}}}\d y_1\cdots \d y_N, \\
I_2^\delta(x)&= \sum_{j\neq l} \int_{B_r(0)}
\frac{\vert f_2(y_j+x_j,y_l+x_l) \vert  }{(y_1^2+\cdots + y_N^2  )^{\frac{3N-2+\delta}{2}}} \d y_1\cdots \d y_N .
\end{align*}
We next demonstrate that
\begin{align}\label{eq:KatoI1}
I_1^\delta(x) &\leq C_N \int_{B_{r;3}(x)} \frac{\vert f_1(y_1) \vert}{\vert y_1-x_1 \vert^{3-2 +\delta}} \d y_1,\\
I_2^\delta(x) & \leq C_N \int_{B_{r;6}(x)} \frac{\vert f_2(y_1,y_2) \vert}{\vert (y_1,y_2)-(x_1,x_2) \vert^{6-2 +\delta}} \d y_1 \d y_2,
\label{eq:KatoI2}
\end{align}
where the second index in the given ball-sets $B_{r;3}(x) \subset \R^3$ and $B_{r;6}(x) \subset \R^6$ refers to the respective dimensionality.

To show \eqref{eq:KatoI1}, set $q=(y_2,\dots,y_N)$ and note that
\begin{align*}
&I_1^\delta(x) \leq N \int_{B_{r;3}\times B_{r;3(N-1)}}
\frac{\vert f_1(y_1+x_1) \vert  }{(y_1^2 +q^2)^{\frac{3N-2+\delta}{2}}} \d y_1 \d q \\
&= C_N \int_{B_{r;3}} \vert f_1(y_1+x_1) \vert 
\left(  \int_0^r \frac{q^{3(N-1)-1} \d q}{(y_1^2+q^2)^{\frac{3N-2 + \delta}{2}}}   \right) \d y_1 \\
& = C_N \int_{B_{r;3}} \frac{\vert f_1(y_1+x_1) \vert}
{\vert y_1 \vert^{3N-2+\delta}} \\
&\quad \quad  \times \left(  \int_0^r \frac{q^{3N-4} \d q}{(1+(q/\vert y_1\vert)^2)^{\frac{3N-2 + \delta}{2}}}   \right) \d y_1 \\
& \leq C_N \, J_1^\delta\int_{B_{r;3}(x_1)} \frac{\vert f_1(y_1) \vert}{\vert y_1-x_1 \vert^{3-2 +\delta}}\d y_1,
\end{align*}
where we have defined the integral
\[
J_1^\delta = \int_0^\infty \frac{s^{3N-4}}
{(1+s^2)^{\frac{3N-2+\delta}{2}}}\d s.
\] 
Now, $J_1^\delta$ is finite since
\begin{align*}
J_1^\delta \leq\int_0^1 s^{3N-4}\d s + \int_1^\infty s^{-2-\delta} \d s <+\infty.
\end{align*}
This establishes \eqref{eq:KatoI1}. The proof of \eqref{eq:KatoI2} is similar and included for the sake of completeness. Set $q=(y_3,\dots,y_N)$, then
\begin{align*}
&I_1 \leq N \int_{B_{r;6}\times B_{r;3(N-2)}}
\frac{\vert u(y_1+x_1,y_2+x_2) \vert  }{(y_1^2+y_2^2  +q^2)^{ \frac{3N-2+\delta}{2}}} \d y_1 \d y_2 \d q \\
&= C_N \int_{B_{r;6}} \vert u(y_1+x_1,y_2+x_2) \vert  \\
&\quad \quad \times \left(  \int_0^r \frac{q^{3(N-2)-1} \d q}{(y_1^2+y_2^2+q^2)^{ \frac{3N-2 + \delta}{2}}}   \right) \d y_1 \d y_2\\
& = C_N \int_{B_{r;6}} \frac{\vert u(y_1+x_1,y_2+x_2) \vert}
{\vert (y_1,y_2) \vert^{3N-2+\delta}} \\
&\quad \quad \times \left(  \int_0^r \frac{q^{3N-7} dq}{(1+(q/\vert (y_1,y_2)\vert)^2)^{ \frac{3N-2 + \delta}{2}  }}   \right) \d y_1 \d y_2 \\
& \leq C_N \int_{B_{r;6}(x_1,x_2)} \frac{\vert u(y_1,y_2) \vert}{\vert (y_1,y_2)-(x_1,x_2) \vert^{6-2 +\delta}} \d y_1 \d y_2 \\
&\quad \quad \times \int_0^\infty \frac{s^{3N-7}}{(1+s^2)^{\frac{3N-2+\delta}{2}}}\d s.
\end{align*}
Corollary~\ref{cor:atoms} now follows from Theorem~\ref{Thm:Main} (Eq.~\eqref{theta} in Assumption~\ref{A:1} is fulfilled by Remark~\ref{rmk:Kurata}). 
\end{proof}

\begin{proof}[Proof of Corollary~\ref{cor:mol-case}]
We first reduce the molecular case to the atomic one. 
Since the UCP from sets of positive measure is local, it can be applied to any open set
in the domain individually. So instead of one singularity (the $y$ of
Assumption~\ref{A:1}), we can treat an arbitrary (yet countable) number of
singularities if they do not have an accumulation point. For this just choose an
open cover $\{U_j \}$ of $\mathbb R^3$ where each $U_j$ contains not more than one
nucleus $x_{\mathrm{nuc};j}$. It remains to show that all $q_{x_{\mathrm{nuc};j}}$, $b_{x_{\mathrm{nuc};j}}$
belong to the respective local Kato classes and we are done if we prove the results for atoms.

In the sequel we let  $v(x_1) = - Z|x_1- x_\mathrm{nuc}|^{-1}$, $ x_\mathrm{nuc}\in\mathbb R^3$, $Z>0$, and 
$u(x_1,x_2)=|x_1-x_2|^{-1}$. In this case $v_-\in L_\mathrm{loc}^{3/2}(\mathbb R^3)$ and with the choice $\underline x_\mathrm{nuc}=( x_\mathrm{nuc},\dots, x_\mathrm{nuc}) \in \mathbb R^D$, we have with $Q_{x_\mathrm{nuc}}(x)  = Q_{\underline x_\mathrm{nuc}}(x) $ the equality
\begin{align*}
Q_{x_\mathrm{nuc}}(x) &= \Big(\sum_j  \frac{-2Z}{|x_j-x_\mathrm{nuc}|} + \sum_{j\neq l} \frac{1}{|x_j - x_l|}\\
&\quad +\sum_j (x_j-x_\mathrm{nuc}) \cdot\nabla_j  \Big(\frac{-Z}{|x_j-x_\mathrm{nuc}|}\Big)   \\
& \quad  +   \frac 1 2\sum_{j\neq l} \big[(x_j-x_\mathrm{nuc})\cdot\nabla_j + (x_l-x_\mathrm{nuc}) \cdot\nabla_l\big]\\
&\quad \quad \times \frac{1}{|(x_j- x_\mathrm{nuc}) - (x_l - x_\mathrm{nuc})|}\Big)_- \\
&= (V(x)+U(x))_- \leq 2V_-(x)\,.
\end{align*}
Thus, in this case we can choose $q_{1,x_\mathrm{nuc}} = v_-$ and $q_{2,x_\mathrm{nuc}} =0$.

Furthermore, for $0<\delta <1$, we claim that $V,U \in K_\mathrm{loc}^{D,\delta}$. By the first part it suffices to show  $v\in K_\mathrm{loc}^{3,\delta}$ and $u\in K_\mathrm{loc}^{6,\delta}$. For $v\in  K_\mathrm{loc}^{3,\delta}$, we introduce polar coordinates with radius $s$ and polar angle $t$. 
Then it holds that $\d y= 2\pi s^2 \sin t \,\d t \d s $ and $y\cdot x = -s\vert x \vert \cos t$. For $f_1(x)=\vert x\vert^{-1}$ it follows that  
\begin{align*}
\eta^K(r;f_1)
=\sup_{ |x|\leq R} \int_0^{r}   
\Big(   \int_{0}^{\pi}    \frac{2\pi  s^{1-\delta} \sin t \,\d t  }{( s^2 - 2  s |x|\cos t + |x|^2)^{1/2} }\Big) \d s .
\end{align*}
We integrate over $t$, use $\big|  s + |x|- | s  - |x||\big| \leq 2 |x|$,  
and the conclusion is obtained for $v$.\\
\indent In a similar fashion, for $u$ we establish that with $f_2(x)= \vert x_1-x_2\vert^{-1}$ 
\begin{align*}
\eta^K(r;f_2) \leq C \int_0^{r} \left( \int_{s_1}^{r} \frac{ s_2 \, \d s_2 }{( s_1^2 + s_2^2)^{2+\frac{2}{\delta}}} \right) s_1^2 \d s_1 
\leq  C r^{1-\delta}
\end{align*}
such that it follows $u\in K_\mathrm{loc}^{6,\delta}$, since
\begin{align*}
	&\int_{B_r(0)} \frac{1}{\vert y_1-y_2 \vert} \frac{1}{\vert y\vert^{6-2+\delta}}\d y_1 \d y_2 \\
 &\leq C \int_0^r \int_0^r  
	\left(\int_0^\pi \frac{2\pi  s_1^2 s_2^2 \sin t \,\d t  }{( s_1^2 - 2  s_1 s_2\cos t + s_2^2)^{1/2} }\right)\\ 
	&\quad \times \frac{\d s_1\d s_2}{(s_1^2+s_2^2)^{2+\frac{2}{\delta}}}\\
	& \leq C \int_0^r \int_0^r \frac{(s_1+s_2 - \vert s_1-s_2\vert) }{(s_1^2+s_2^2)^{2+\frac{2}{\delta}}}s_1s_2 \d s_1\d s_2 \\
	&\leq C \int_0^r \left[  \int_0^{s_1} \frac{s_1 s_2^2 \d s_2}{(s_1^2+s_2^2)^{2+\frac{2}{\delta}}} + \int_{s_1}^r \frac{s_1^2 s_2 \d s_2}{(s_1^2+s_2^2)^{2+\frac{2}{\delta}}}    \right] \d s_1 \\
	& \leq C \int_0^r\int_{s_1}^r \frac{ s_2 \d s_2}{(s_1^2+s_2^2)^{2+\frac{2}{\delta}}} s_1^2\d s_1 \leq Cr^{1-\delta}.
\end{align*}
The atomic case is now a consequence of Corollary~\ref{cor:atoms}.
\end{proof}

\section{Conclusion}

In this work we were able to show the unique-continuation property from sets of positive measures for the important case of the many-body magnetic Schr\"odinger equation for classes of potentials that are independent of the particle number. This is crucial in order to not artificially restrict the permitted potentials in large systems. We further specifically addressed molecular Hamiltonians, thus covering most cases that usually arise in physics.

\section*{Acknowledgements}

AL is grateful for the hospitality received at the Max Planck Institute for the Structure and Dynamics of Matter while visiting MP. The authors want to thank Louis Garrigue for useful discussion that led
to the inclusion of the molecular case, and Fabian Faulstich for comments and suggestions that improved the manuscript.

\section*{Funding Information}
The work of AL was supported by the Norwegian
Research Council through the CoE 
Hylleraas Centre for Quantum Molecular Sciences Grant No. 262695 and 
through the European Research Council, ERC-STG-2014 Grant Agreement
No.~639508. The work of MB was supported by the Swedish Research Council
(VR), Grant No.~2016-05482. 
The work of MP was supported by the Erwin Schr\"odinger
Fellowship J 4107-N27 of the FWF (Austrian Science Fund).

\bibliography{refs}

\begin{thebibliography}{36}%
\makeatletter
\providecommand \@ifxundefined [1]{%
 \@ifx{#1\undefined}
}%
\providecommand \@ifnum [1]{%
 \ifnum #1\expandafter \@firstoftwo
 \else \expandafter \@secondoftwo
 \fi
}%
\providecommand \@ifx [1]{%
 \ifx #1\expandafter \@firstoftwo
 \else \expandafter \@secondoftwo
 \fi
}%
\providecommand \natexlab [1]{#1}%
\providecommand \enquote  [1]{``#1''}%
\providecommand \bibnamefont  [1]{#1}%
\providecommand \bibfnamefont [1]{#1}%
\providecommand \citenamefont [1]{#1}%
\providecommand \href@noop [0]{\@secondoftwo}%
\providecommand \href [0]{\begingroup \@sanitize@url \@href}%
\providecommand \@href[1]{\@@startlink{#1}\@@href}%
\providecommand \@@href[1]{\endgroup#1\@@endlink}%
\providecommand \@sanitize@url [0]{\catcode `\\12\catcode `\$12\catcode
  `\&12\catcode `\#12\catcode `\^12\catcode `\_12\catcode `\%12\relax}%
\providecommand \@@startlink[1]{}%
\providecommand \@@endlink[0]{}%
\providecommand \url  [0]{\begingroup\@sanitize@url \@url }%
\providecommand \@url [1]{\endgroup\@href {#1}{\urlprefix }}%
\providecommand \urlprefix  [0]{URL }%
\providecommand \Eprint [0]{\href }%
\providecommand \doibase [0]{http://dx.doi.org/}%
\providecommand \selectlanguage [0]{\@gobble}%
\providecommand \bibinfo  [0]{\@secondoftwo}%
\providecommand \bibfield  [0]{\@secondoftwo}%
\providecommand \translation [1]{[#1]}%
\providecommand \BibitemOpen [0]{}%
\providecommand \bibitemStop [0]{}%
\providecommand \bibitemNoStop [0]{.\EOS\space}%
\providecommand \EOS [0]{\spacefactor3000\relax}%
\providecommand \BibitemShut  [1]{\csname bibitem#1\endcsname}%
\let\auto@bib@innerbib\@empty
\bibitem [{\citenamefont {Hohenberg}\ and\ \citenamefont
  {Kohn}(1964)}]{Hohenberg1964}%
  \BibitemOpen
  \bibfield  {author} {\bibinfo {author} {\bibfnamefont {P.}~\bibnamefont
  {Hohenberg}}\ and\ \bibinfo {author} {\bibfnamefont {W.}~\bibnamefont
  {Kohn}},\ }\href {\doibase 10.1103/PhysRev.136.B864} {\bibfield  {journal}
  {\bibinfo  {journal} {Phys. Rev.}\ }\textbf {\bibinfo {volume} {136}},\
  \bibinfo {pages} {B864} (\bibinfo {year} {1964})}\BibitemShut {NoStop}%
\bibitem [{\citenamefont {Englisch}\ and\ \citenamefont
  {Englisch}(1983)}]{ENGLISCH1983}%
  \BibitemOpen
  \bibfield  {author} {\bibinfo {author} {\bibfnamefont {H.}~\bibnamefont
  {Englisch}}\ and\ \bibinfo {author} {\bibfnamefont {R.}~\bibnamefont
  {Englisch}},\ }\href {\doibase 10.1016/0378-4371(83)90254-6} {\bibfield
  {journal} {\bibinfo  {journal} {Physica A Stat. Mech. Appl.}\ }\textbf
  {\bibinfo {volume} {121}},\ \bibinfo {pages} {253 } (\bibinfo {year}
  {1983})}\BibitemShut {NoStop}%
\bibitem [{\citenamefont {Burke}(2012)}]{Burke2012}%
  \BibitemOpen
  \bibfield  {author} {\bibinfo {author} {\bibfnamefont {K.}~\bibnamefont
  {Burke}},\ }\href {\doibase 10.1063/1.4704546} {\bibfield  {journal}
  {\bibinfo  {journal} {J. Chem. Phys.}\ }\textbf {\bibinfo {volume} {136}},\
  \bibinfo {pages} {150901} (\bibinfo {year} {2012})}\BibitemShut {NoStop}%
\bibitem [{\citenamefont {Jones}(2015)}]{Jones2015}%
  \BibitemOpen
  \bibfield  {author} {\bibinfo {author} {\bibfnamefont {R.~O.}\ \bibnamefont
  {Jones}},\ }\href {\doibase 10.1103/RevModPhys.87.897} {\bibfield  {journal}
  {\bibinfo  {journal} {Rev. Mod. Phys.}\ }\textbf {\bibinfo {volume} {87}},\
  \bibinfo {pages} {897} (\bibinfo {year} {2015})}\BibitemShut {NoStop}%
\bibitem [{\citenamefont {Vignale}\ and\ \citenamefont
  {Rasolt}(1987)}]{Vignale1987}%
  \BibitemOpen
  \bibfield  {author} {\bibinfo {author} {\bibfnamefont {G.}~\bibnamefont
  {Vignale}}\ and\ \bibinfo {author} {\bibfnamefont {M.}~\bibnamefont
  {Rasolt}},\ }\href {\doibase 10.1103/PhysRevLett.59.2360} {\bibfield
  {journal} {\bibinfo  {journal} {Phys. Rev. Lett.}\ }\textbf {\bibinfo
  {volume} {59}},\ \bibinfo {pages} {2360} (\bibinfo {year}
  {1987})}\BibitemShut {NoStop}%
\bibitem [{\citenamefont {Vignale}\ and\ \citenamefont
  {Rasolt}(1988)}]{Vignale1988}%
  \BibitemOpen
  \bibfield  {author} {\bibinfo {author} {\bibfnamefont {G.}~\bibnamefont
  {Vignale}}\ and\ \bibinfo {author} {\bibfnamefont {M.}~\bibnamefont
  {Rasolt}},\ }\href {\doibase 10.1103/PhysRevB.37.10685} {\bibfield  {journal}
  {\bibinfo  {journal} {Phys. Rev. B}\ }\textbf {\bibinfo {volume} {37}},\
  \bibinfo {pages} {10685} (\bibinfo {year} {1988})}\BibitemShut {NoStop}%
\bibitem [{\citenamefont {Diener}(1991)}]{Diener}%
  \BibitemOpen
  \bibfield  {author} {\bibinfo {author} {\bibfnamefont {G.}~\bibnamefont
  {Diener}},\ }\href {\doibase 10.1088/0953-8984/3/47/014} {\bibfield
  {journal} {\bibinfo  {journal} {J. Phys.: Condens. Matter}\ }\textbf
  {\bibinfo {volume} {3}},\ \bibinfo {pages} {9417} (\bibinfo {year}
  {1991})}\BibitemShut {NoStop}%
\bibitem [{\citenamefont {Capelle}\ and\ \citenamefont
  {Vignale}(2002)}]{Capelle2002}%
  \BibitemOpen
  \bibfield  {author} {\bibinfo {author} {\bibfnamefont {K.}~\bibnamefont
  {Capelle}}\ and\ \bibinfo {author} {\bibfnamefont {G.}~\bibnamefont
  {Vignale}},\ }\href {\doibase 10.1103/PhysRevB.65.113106} {\bibfield
  {journal} {\bibinfo  {journal} {Phys. Rev. B}\ }\textbf {\bibinfo {volume}
  {65}},\ \bibinfo {pages} {113106} (\bibinfo {year} {2002})}\BibitemShut
  {NoStop}%
\bibitem [{\citenamefont {Laestadius}\ and\ \citenamefont
  {Benedicks}(2014)}]{Laestadius2014}%
  \BibitemOpen
  \bibfield  {author} {\bibinfo {author} {\bibfnamefont {A.}~\bibnamefont
  {Laestadius}}\ and\ \bibinfo {author} {\bibfnamefont {M.}~\bibnamefont
  {Benedicks}},\ }\href {\doibase 10.1002/qua.24668} {\bibfield  {journal}
  {\bibinfo  {journal} {Int. J. Quantum Chem.}\ }\textbf {\bibinfo {volume}
  {114}},\ \bibinfo {pages} {782} (\bibinfo {year} {2014})}\BibitemShut
  {NoStop}%
\bibitem [{\citenamefont {Laestadius}\ and\ \citenamefont
  {Tellgren}(2018)}]{LaestadiusTellgren2018}%
  \BibitemOpen
  \bibfield  {author} {\bibinfo {author} {\bibfnamefont {A.}~\bibnamefont
  {Laestadius}}\ and\ \bibinfo {author} {\bibfnamefont {E.~I.}\ \bibnamefont
  {Tellgren}},\ }\href {\doibase 10.1103/PhysRevA.97.022514} {\bibfield
  {journal} {\bibinfo  {journal} {Phys. Rev. A}\ }\textbf {\bibinfo {volume}
  {97}},\ \bibinfo {pages} {022514} (\bibinfo {year} {2018})}\BibitemShut
  {NoStop}%
\bibitem [{\citenamefont {Tellgren}\ \emph {et~al.}(2012)\citenamefont
  {Tellgren}, \citenamefont {Kvaal}, \citenamefont {Sagvolden}, \citenamefont
  {Ekstr\"om}, \citenamefont {Teale},\ and\ \citenamefont
  {Helgaker}}]{Tellgren2012}%
  \BibitemOpen
  \bibfield  {author} {\bibinfo {author} {\bibfnamefont {E.~I.}\ \bibnamefont
  {Tellgren}}, \bibinfo {author} {\bibfnamefont {S.}~\bibnamefont {Kvaal}},
  \bibinfo {author} {\bibfnamefont {E.}~\bibnamefont {Sagvolden}}, \bibinfo
  {author} {\bibfnamefont {U.}~\bibnamefont {Ekstr\"om}}, \bibinfo {author}
  {\bibfnamefont {A.~M.}\ \bibnamefont {Teale}}, \ and\ \bibinfo {author}
  {\bibfnamefont {T.}~\bibnamefont {Helgaker}},\ }\href {\doibase
  10.1103/PhysRevA.86.062506} {\bibfield  {journal} {\bibinfo  {journal} {Phys.
  Rev. A}\ }\textbf {\bibinfo {volume} {86}},\ \bibinfo {pages} {062506}
  (\bibinfo {year} {2012})}\BibitemShut {NoStop}%
\bibitem [{\citenamefont {Tellgren}\ \emph {et~al.}(2018)\citenamefont
  {Tellgren}, \citenamefont {Laestadius}, \citenamefont {Helgaker},
  \citenamefont {Kvaal},\ and\ \citenamefont {Teale}}]{Tellgren2018}%
  \BibitemOpen
  \bibfield  {author} {\bibinfo {author} {\bibfnamefont {E.~I.}\ \bibnamefont
  {Tellgren}}, \bibinfo {author} {\bibfnamefont {A.}~\bibnamefont
  {Laestadius}}, \bibinfo {author} {\bibfnamefont {T.}~\bibnamefont
  {Helgaker}}, \bibinfo {author} {\bibfnamefont {S.}~\bibnamefont {Kvaal}}, \
  and\ \bibinfo {author} {\bibfnamefont {A.~M.}\ \bibnamefont {Teale}},\ }\href
  {\doibase 10.1063/1.5007300} {\bibfield  {journal} {\bibinfo  {journal} {J.
  Chem. Phys.}\ }\textbf {\bibinfo {volume} {148}},\ \bibinfo {pages} {024101}
  (\bibinfo {year} {2018})}\BibitemShut {NoStop}%
\bibitem [{\citenamefont {Ruggenthaler}(2017)}]{Ruggenthaler2015}%
  \BibitemOpen
  \bibfield  {author} {\bibinfo {author} {\bibfnamefont {M.}~\bibnamefont
  {Ruggenthaler}},\ }\href@noop {} {\  (\bibinfo {year} {2017})},\ \Eprint
  {http://arxiv.org/abs/1509.01417} {arXiv:1509.01417 [quant-ph]} \BibitemShut
  {NoStop}%
\bibitem [{\citenamefont {Tellgren}(2018)}]{TellgrenSolo2018}%
  \BibitemOpen
  \bibfield  {author} {\bibinfo {author} {\bibfnamefont {E.~I.}\ \bibnamefont
  {Tellgren}},\ }\href {\doibase 10.1103/PhysRevA.97.012504} {\bibfield
  {journal} {\bibinfo  {journal} {Phys. Rev. A}\ }\textbf {\bibinfo {volume}
  {97}},\ \bibinfo {pages} {012504} (\bibinfo {year} {2018})}\BibitemShut
  {NoStop}%
\bibitem [{\citenamefont {Garrigue}(2019{\natexlab{a}})}]{Garrigue2019}%
  \BibitemOpen
  \bibfield  {author} {\bibinfo {author} {\bibfnamefont {L.}~\bibnamefont
  {Garrigue}},\ }\href@noop {} {\  (\bibinfo {year} {2019}{\natexlab{a}})},\
  \Eprint {http://arxiv.org/abs/1901.03207} {arXiv:1901.03207 [math-ph]}
  \BibitemShut {NoStop}%
\bibitem [{\citenamefont {Garrigue}(2019{\natexlab{b}})}]{Garrigue2019b}%
  \BibitemOpen
  \bibfield  {author} {\bibinfo {author} {\bibfnamefont {L.}~\bibnamefont
  {Garrigue}},\ }\href {\doibase 10.1007/s10955-019-02365-6} {\bibfield
  {journal} {\bibinfo  {journal} {Journal of Statistical Physics}\ }\textbf
  {\bibinfo {volume} {177}},\ \bibinfo {pages} {415} (\bibinfo {year}
  {2019}{\natexlab{b}})}\BibitemShut {NoStop}%
\bibitem [{\citenamefont {Kvaal}\ \emph {et~al.}(2014)\citenamefont {Kvaal},
  \citenamefont {Ekstr{\"o}m}, \citenamefont {Teale},\ and\ \citenamefont
  {Helgaker}}]{Kvaal2014}%
  \BibitemOpen
  \bibfield  {author} {\bibinfo {author} {\bibfnamefont {S.}~\bibnamefont
  {Kvaal}}, \bibinfo {author} {\bibfnamefont {U.}~\bibnamefont {Ekstr{\"o}m}},
  \bibinfo {author} {\bibfnamefont {A.~M.}\ \bibnamefont {Teale}}, \ and\
  \bibinfo {author} {\bibfnamefont {T.}~\bibnamefont {Helgaker}},\ }\href
  {\doibase 10.1063/1.4867005} {\bibfield  {journal} {\bibinfo  {journal} {J.
  Chem. Phys.}\ }\textbf {\bibinfo {volume} {140}},\ \bibinfo {pages} {18A518}
  (\bibinfo {year} {2014})}\BibitemShut {NoStop}%
\bibitem [{\citenamefont {Laestadius}\ \emph {et~al.}(2018)\citenamefont
  {Laestadius}, \citenamefont {Penz}, \citenamefont {Tellgren}, \citenamefont
  {Ruggenthaler}, \citenamefont {Kvaal},\ and\ \citenamefont
  {Helgaker}}]{KSpaper2018}%
  \BibitemOpen
  \bibfield  {author} {\bibinfo {author} {\bibfnamefont {A.}~\bibnamefont
  {Laestadius}}, \bibinfo {author} {\bibfnamefont {M.}~\bibnamefont {Penz}},
  \bibinfo {author} {\bibfnamefont {E.~I.}\ \bibnamefont {Tellgren}}, \bibinfo
  {author} {\bibfnamefont {M.}~\bibnamefont {Ruggenthaler}}, \bibinfo {author}
  {\bibfnamefont {S.}~\bibnamefont {Kvaal}}, \ and\ \bibinfo {author}
  {\bibfnamefont {T.}~\bibnamefont {Helgaker}},\ }\href {\doibase
  10.1063/1.5037790} {\bibfield  {journal} {\bibinfo  {journal} {J. Chem.
  Phys.}\ }\textbf {\bibinfo {volume} {149}},\ \bibinfo {pages} {164103}
  (\bibinfo {year} {2018})}\BibitemShut {NoStop}%
\bibitem [{\citenamefont {Laestadius}\ \emph {et~al.}(2019)\citenamefont
  {Laestadius}, \citenamefont {Penz}, \citenamefont {Tellgren}, \citenamefont
  {Ruggenthaler}, \citenamefont {Kvaal},\ and\ \citenamefont
  {Helgaker}}]{MY-CDFTpaper2019}%
  \BibitemOpen
  \bibfield  {author} {\bibinfo {author} {\bibfnamefont {A.}~\bibnamefont
  {Laestadius}}, \bibinfo {author} {\bibfnamefont {M.}~\bibnamefont {Penz}},
  \bibinfo {author} {\bibfnamefont {E.~I.}\ \bibnamefont {Tellgren}}, \bibinfo
  {author} {\bibfnamefont {M.}~\bibnamefont {Ruggenthaler}}, \bibinfo {author}
  {\bibfnamefont {S.}~\bibnamefont {Kvaal}}, \ and\ \bibinfo {author}
  {\bibfnamefont {T.}~\bibnamefont {Helgaker}},\ }\href {\doibase
  10.1021/acs.jctc.9b00141} {\bibfield  {journal} {\bibinfo  {journal} {J.
  Chem. Theory Comput.}\ }\textbf {\bibinfo {volume} {15}},\ \bibinfo {pages}
  {4003} (\bibinfo {year} {2019})}\BibitemShut {NoStop}%
\bibitem [{\citenamefont {Penz}\ \emph {et~al.}(2019)\citenamefont {Penz},
  \citenamefont {Laestadius}, \citenamefont {Tellgren},\ and\ \citenamefont
  {Ruggenthaler}}]{penz2019guaranteed}%
  \BibitemOpen
  \bibfield  {author} {\bibinfo {author} {\bibfnamefont {M.}~\bibnamefont
  {Penz}}, \bibinfo {author} {\bibfnamefont {A.}~\bibnamefont {Laestadius}},
  \bibinfo {author} {\bibfnamefont {E.~I.}\ \bibnamefont {Tellgren}}, \ and\
  \bibinfo {author} {\bibfnamefont {M.}~\bibnamefont {Ruggenthaler}},\ }\href
  {https://doi.org/10.1103/physrevlett.123.037401} {\bibfield  {journal}
  {\bibinfo  {journal} {Phys. Rev. Lett.}\ }\textbf {\bibinfo {volume} {123}},\
  \bibinfo {pages} {037401} (\bibinfo {year} {2019})}\BibitemShut {NoStop}%
\bibitem [{\citenamefont {H\"ormander}(1983)}]{H}%
  \BibitemOpen
  \bibfield  {author} {\bibinfo {author} {\bibfnamefont {L.}~\bibnamefont
  {H\"ormander}},\ }\href {\doibase 10.1080/03605308308820262} {\bibfield
  {journal} {\bibinfo  {journal} {Comm. in P.D.E.}\ }\textbf {\bibinfo {volume}
  {8}},\ \bibinfo {pages} {21} (\bibinfo {year} {1983})}\BibitemShut {NoStop}%
\bibitem [{\citenamefont {Barcelo}\ \emph {et~al.}(1988)\citenamefont
  {Barcelo}, \citenamefont {Kenig}, \citenamefont {Ruiz},\ and\ \citenamefont
  {Sogge}}]{BKRS}%
  \BibitemOpen
  \bibfield  {author} {\bibinfo {author} {\bibfnamefont {B.}~\bibnamefont
  {Barcelo}}, \bibinfo {author} {\bibfnamefont {C.~E.}\ \bibnamefont {Kenig}},
  \bibinfo {author} {\bibfnamefont {A.}~\bibnamefont {Ruiz}}, \ and\ \bibinfo
  {author} {\bibfnamefont {C.~D.}\ \bibnamefont {Sogge}},\ }\href {\doibase
  10.1215/ijm/1255989128} {\bibfield  {journal} {\bibinfo  {journal} {Illinois
  J. Math.}\ }\textbf {\bibinfo {volume} {32}},\ \bibinfo {pages} {230}
  (\bibinfo {year} {1988})}\BibitemShut {NoStop}%
\bibitem [{\citenamefont {Wolff}(1992)}]{Wolff1992}%
  \BibitemOpen
  \bibfield  {author} {\bibinfo {author} {\bibfnamefont {T.~H.}\ \bibnamefont
  {Wolff}},\ }\href {\doibase 10.1007/BF01896975} {\bibfield  {journal}
  {\bibinfo  {journal} {Geom. Funct. Anal.}\ }\textbf {\bibinfo {volume} {2}},\
  \bibinfo {pages} {225} (\bibinfo {year} {1992})}\BibitemShut {NoStop}%
\bibitem [{\citenamefont {Kurata}(1993)}]{Kurata1}%
  \BibitemOpen
  \bibfield  {author} {\bibinfo {author} {\bibfnamefont {K.}~\bibnamefont
  {Kurata}},\ }\href {\doibase 10.1080/03605309308820968} {\bibfield  {journal}
  {\bibinfo  {journal} {Comm. in P.D.E.}\ }\textbf {\bibinfo {volume} {18}},\
  \bibinfo {pages} {1161} (\bibinfo {year} {1993})}\BibitemShut {NoStop}%
\bibitem [{\citenamefont {Kurata}(1997)}]{Kurata2}%
  \BibitemOpen
  \bibfield  {author} {\bibinfo {author} {\bibfnamefont {K.}~\bibnamefont
  {Kurata}},\ }\href {\doibase 10.1090/S0002-9939-97-03672-1} {\bibfield
  {journal} {\bibinfo  {journal} {Proc. Amer. Math. Soc.}\ }\textbf {\bibinfo
  {volume} {125}},\ \bibinfo {pages} {853} (\bibinfo {year}
  {1997})}\BibitemShut {NoStop}%
\bibitem [{\citenamefont {Regbaoui}(2001)}]{Regbaoui}%
  \BibitemOpen
  \bibfield  {author} {\bibinfo {author} {\bibfnamefont {R.}~\bibnamefont
  {Regbaoui}},\ }\href@noop {} {\emph {\bibinfo {title} {Carleman Estimates and
  Applications to Uniqueness and Control Theory}}}\ (\bibinfo  {publisher}
  {Birkhäuser, Basel},\ \bibinfo {year} {2001})\ pp.\ \bibinfo {pages}
  {179--190}\BibitemShut {NoStop}%
\bibitem [{\citenamefont {Lieb}\ and\ \citenamefont {Loss}(2001)}]{LiebLoss}%
  \BibitemOpen
  \bibfield  {author} {\bibinfo {author} {\bibfnamefont {E.~H.}\ \bibnamefont
  {Lieb}}\ and\ \bibinfo {author} {\bibfnamefont {M.}~\bibnamefont {Loss}},\
  }\href@noop {} {\emph {\bibinfo {title} {Analysis}}}\ (\bibinfo  {publisher}
  {American Mathematical Society, Providence, Rhode Island},\ \bibinfo {year}
  {2001})\BibitemShut {NoStop}%
\bibitem [{\citenamefont {Laestadius}(2014)}]{Laestadius2014b}%
  \BibitemOpen
  \bibfield  {author} {\bibinfo {author} {\bibfnamefont {A.}~\bibnamefont
  {Laestadius}},\ }\href {\doibase 10.1007/s10910-014-0400-7} {\bibfield
  {journal} {\bibinfo  {journal} {J. Math. Chem.}\ }\textbf {\bibinfo {volume}
  {52}},\ \bibinfo {pages} {2581} (\bibinfo {year} {2014})}\BibitemShut
  {NoStop}%
\bibitem [{\citenamefont {Garrigue}(2018)}]{Garrigue2018}%
  \BibitemOpen
  \bibfield  {author} {\bibinfo {author} {\bibfnamefont {L.}~\bibnamefont
  {Garrigue}},\ }\href {\doibase 10.1007/s11040-018-9287-z} {\bibfield
  {journal} {\bibinfo  {journal} {Math. Phys. Anal. Geom.}\ }\textbf {\bibinfo
  {volume} {21}},\ \bibinfo {pages} {27} (\bibinfo {year} {2018})}\BibitemShut
  {NoStop}%
\bibitem [{\citenamefont {Lammert}(2018)}]{Lammert2018}%
  \BibitemOpen
  \bibfield  {author} {\bibinfo {author} {\bibfnamefont {P.~E.}\ \bibnamefont
  {Lammert}},\ }\href {\doibase 10.1063/1.5034215} {\bibfield  {journal}
  {\bibinfo  {journal} {J. Math. Phys.}\ }\textbf {\bibinfo {volume} {59}},\
  \bibinfo {pages} {042110} (\bibinfo {year} {2018})}\BibitemShut {NoStop}%
\bibitem [{\citenamefont {Lieb}(1983)}]{Lieb1983}%
  \BibitemOpen
  \bibfield  {author} {\bibinfo {author} {\bibfnamefont {E.~H.}\ \bibnamefont
  {Lieb}},\ }\href {\doibase 10.1002/qua.560240302} {\bibfield  {journal}
  {\bibinfo  {journal} {Int. J. Quantum Chem.}\ }\textbf {\bibinfo {volume}
  {24}},\ \bibinfo {pages} {243} (\bibinfo {year} {1983})}\BibitemShut
  {NoStop}%
\bibitem [{\citenamefont {Koch}\ and\ \citenamefont
  {Tataru}(2002)}]{Koch-Tataru2007}%
  \BibitemOpen
  \bibfield  {author} {\bibinfo {author} {\bibfnamefont {H.}~\bibnamefont
  {Koch}}\ and\ \bibinfo {author} {\bibfnamefont {D.}~\bibnamefont {Tataru}},\
  }\href@noop {} {\bibfield  {journal} {\bibinfo  {journal} {J. Reine Angew.
  Math.}\ }\textbf {\bibinfo {volume} {542}},\ \bibinfo {pages} {133} (\bibinfo
  {year} {2002})},\ \bibinfo {note} {an updated version from 2007 is available
  on Daniel Tataru’s homepage}\BibitemShut {NoStop}%
\bibitem [{\citenamefont {Grayce}\ and\ \citenamefont {Harris}(1994)}]{GRAYCE}%
  \BibitemOpen
  \bibfield  {author} {\bibinfo {author} {\bibfnamefont {C.~J.}\ \bibnamefont
  {Grayce}}\ and\ \bibinfo {author} {\bibfnamefont {R.~A.}\ \bibnamefont
  {Harris}},\ }\href {\doibase 10.1103/PhysRevA.50.3089} {\bibfield  {journal}
  {\bibinfo  {journal} {Phys. Rev. A}\ }\textbf {\bibinfo {volume} {50}},\
  \bibinfo {pages} {3089} (\bibinfo {year} {1994})}\BibitemShut {NoStop}%
\bibitem [{\citenamefont {Ladyzenskaya}\ and\ \citenamefont
  {Ural'tzeva}(1968)}]{Ladyzenskaya1968}%
  \BibitemOpen
  \bibfield  {author} {\bibinfo {author} {\bibfnamefont {O.~A.}\ \bibnamefont
  {Ladyzenskaya}}\ and\ \bibinfo {author} {\bibfnamefont {N.~N.}\ \bibnamefont
  {Ural'tzeva}},\ }\href@noop {} {\emph {\bibinfo {title} {Linear and
  Quasilinear Elliptic Equations}}}\ (\bibinfo  {publisher} {Academic Press,
  New York},\ \bibinfo {year} {1968})\BibitemShut {NoStop}%
\bibitem [{\citenamefont {Simon}(1982)}]{Simon1982}%
  \BibitemOpen
  \bibfield  {author} {\bibinfo {author} {\bibfnamefont {B.}~\bibnamefont
  {Simon}},\ }\href {\doibase 10.1090/S0273-0979-1982-15041-8} {\bibfield
  {journal} {\bibinfo  {journal} {Bull. Amer. Math. Soc. (N.S.)}\ }\textbf
  {\bibinfo {volume} {7}},\ \bibinfo {pages} {447} (\bibinfo {year}
  {1982})}\BibitemShut {NoStop}%
\bibitem [{\citenamefont {Aizenman}\ and\ \citenamefont
  {Simon}(1982)}]{Aizenman}%
  \BibitemOpen
  \bibfield  {author} {\bibinfo {author} {\bibfnamefont {M.}~\bibnamefont
  {Aizenman}}\ and\ \bibinfo {author} {\bibfnamefont {B.}~\bibnamefont
  {Simon}},\ }\href {\doibase 10.1002/cpa.3160350206} {\bibfield  {journal}
  {\bibinfo  {journal} {Commun. Pure Appl. Math}\ }\textbf {\bibinfo {volume}
  {35}},\ \bibinfo {pages} {209} (\bibinfo {year} {1982})}\BibitemShut
  {NoStop}%
\end{thebibliography}%

\end{document}